\documentclass[letterpaper,conference]{IEEEtran}

\usepackage{ifthen}
\usepackage[cmex10]{amsmath} 

\usepackage{amscd,amssymb}
\usepackage{mathrsfs,fancyhdr,graphics,pifont}
\usepackage{floatflt,subfigure,epsfig,moreverb,multicol,lscape}
\usepackage{supertabular,tabularx,multirow,bar,amssymb}
\usepackage{theorem,psfrag,afterpage,curves,epic,bbm,bbold,color}
\usepackage{wasysym}
\usepackage{pifont}
\usepackage{stmaryrd}
\usepackage[english]{babel}

\newcommand{\ignore}[1]{}

\newcommand{\Inttwo}{\mathbb{Z}_2}
\newcommand{\Zp}{\mathbb{Z}_p}

\newcommand{\R}{\mathbb{R}}

\newcommand{\matr}[1]{\mathbf{#1}}
\newcommand{\vect}[1]{\mathbf{#1}}
\newcommand{\code}[1]{\mathcal{#1}}

\newcommand{\set}[1]{\mathcal{#1}}

\newcommand{\GF}[1]{\mathbb{F}_{#1}}
\newcommand{\GFfour}{\GF{4}}

\newcommand{\defeq}{\triangleq}

\newcommand{\oomega}{\overline{\omega}}

\newcommand{\vb}{\vect{b}}

\newcommand{\vm}{\vect{m}}

\newcommand{\vt}{\vect{t}}

\newcommand{\ox}{\overline{x}}

\newcommand{\tr}{\mathrm{T}}

\newcommand{\innerprod}[2]{\langle #1, #2 \rangle}

\newcommand{\nset}[1]{\mathsf{#1}}

\newcommand{\perps}{{\perp}}

\newtheorem{Lemma}{Lemma}
\newtheorem{Theorem}[Lemma]{Theorem}
\newtheorem{Proposition}[Lemma]{Proposition}
\newtheorem{Corollary}[Lemma]{Corollary}
\newtheorem{Assumption}[Lemma]{Assumption}

\theoremstyle{plain}

\newtheorem{PreDefinition}[Lemma]{{\textbf{Definition}}}
  \newenvironment{Definition}%
    {\begin{PreDefinition}}{\hfill$\square$\end{PreDefinition}}

\theoremstyle{plain}

\newtheorem{PreRemark}[Lemma]{{\textbf{Remark}}}
    {\begin{PreRemark}\upshape}{\hfill$\square$\end{PreRemark}}

\newtheorem{PreExample}[Lemma]{{\textbf{Example}}}
  \newenvironment{Example}%
    {\begin{PreExample}\upshape}{\hfill$\square$\end{PreExample}}

  {\noindent \emph{Proof:}}{\hfill$\square$}

\begin{document}

\title{Stabilizer Quantum Codes: A Unified View \\
       based on Forney-style Factor Graphs}

\author{%
        \authorblockN{Pascal O. Vontobel}
        \authorblockA{Hewlett--Packard Laboratories\\
                      Palo Alto, CA 94304, USA\\
                      Email: pascal.vontobel@ieee.org}
}

\maketitle

\begin{abstract}
  Quantum error-correction codes (QECCs) are a vital ingredient of quantum
  computation and communication systems. In that context it is highly
  desirable to design QECCs that can be represented by graphical models which
  possess a structure that enables efficient and close-to-optimal iterative
  decoding.

  In this paper we focus on stabilizer QECCs, a class of QECCs whose
  construction is rendered non-trivial by the fact that the stabilizer label
  code, a code that is associated with a stabilizer QECC, has to satisfy a
  certain self-orthogonality condition. In order to design graphical models of
  stabilizer label codes that satisfy this condition, we extend a duality
  result for Forney-style factor graphs (FFGs) to the stabilizer label code
  framework. This allows us to formulate a simple FFG design rule for
  constructing stabilizer label codes, a design rule that unifies several
  earlier stabilizer label code constructions.
\end{abstract}

\section{Introduction}
\label{sec:introduction:1}

Graphical models have played a very important role in the recent history of
error-correction coding (ECC) schemes for conventional channel and storage
setups~\cite{Kschischang:Frey:Loeliger:01}. In particular, some of the most
powerful ECC schemes known today, like message-passing iterative (MPI)
decoding of low-density parity-check (LDPC) and turbo codes, can be
represented by graphical models. It is therefore highly desirable to extend
the design and analysis lessons that have been learned from these ECC systems
to quantum error-correction code (QECC) systems, in particular to stabilizer
QECC systems.

For background material and a history of stabilizer QECCs in particular, and
quantum information processing (QIP) in general, we refer to the excellent
textbook by Nielsen and Chuang~\cite{Nielsen:Chuang:00:1}. Alternatively, one
can consult some early papers on stabilizer QECCs, e.g.~\cite{Gottesman:97:1,
  Calderbank:Rains:Shor:Sloane:98:1}, or more recent accounts,
e.g.~\cite{MacKay:Mitchison:McFadden:04:1, Forney:Grassl:Guha:07:1,
  Poulin:Tillich:Ollivier:07:1:subm}. The aim of the present paper is to
introduce a Forney-style factor graph (FFG) framework that allows one to
construct FFGs that represent interesting classes of stabilizer QECCs, more
precisely, that represent interesting classes of stabilizer label codes and
normalizer label codes.  Anyone familiar with the basics of stabilizer QECCs
can then easily formulate the corresponding stabilizer QECCs. (Note that due
to space constraints this paper does not motivate stabilizer QECCs and does
not define them. However, the paper does not use any QIP jargon and should
therefore be accessible to anyone familiar with the basics of coding theory.)

This paper is structured as follows. In Section~\ref{sec:ffgs:1} we introduce
the basics on FFGs and in Section~\ref{sec:dualizing:ffgs:1} we extend a
well-known duality result for FFGs. In
Section~\ref{sec:stabilizer:and:normalizer:label:codes:1} we then show how
this duality result can be used to construct stabilizer and normalizer label
codes. Afterwards, in Section~\ref{sec:examples:1} we discuss several examples
of such codes, in particular we show that our FFG framework unifies earlier
proposed code constructions. We conclude by briefly commenting on
message-passing iterative decoding and linear programming decoding in
Section~\ref{sec:mpi:and:lp:decoding:1}. Because of space limitations we
decided to formulate many of the concepts and results in terms of examples;
most of them can be suitably generalized.

Our notation is quite standard. In particular, the field of real numbers will
be denoted by $\R$ and the ring of integers modulo $p$ by $\Zp$. (If $p$ is a
prime, then $\Zp$ is a field.) The Galois field $\GFfour$ will be based on the
set $\{ 0, 1, \omega, \omega^2 \}$, where $\omega$ satisfies $\omega^2 =
\omega + 1$ (and therefore also $\omega^3 = 1$), and where conjugation is
defined by $\ox = x^2$. Moreover, for any statement $S$ we will use Iverson's
convention which says that $[S] = 1$ if $S$ is true and $[S] = 0$ otherwise.

\section{Forney-style Factor Graphs (FFGs)}
\label{sec:ffgs:1}

\begin{figure}
  \begin{center}
    \epsfig{file=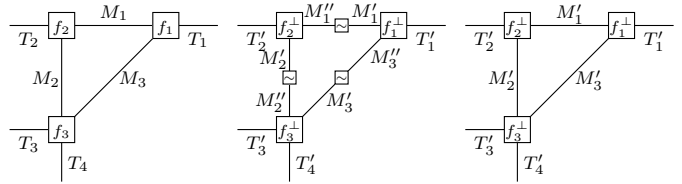, width=\linewidth}
  \end{center}
  \caption{Left: A simple FFG. Middle: Dual of the FFG on the left. Right:
    Dual of the FFG on the left if alphabet groups have characteristic $2$.}
  \label{fig:simple:ffg:1}
\end{figure}

FFGs~\cite{Forney:01:1, Loeliger:04:1, Kschischang:Frey:Loeliger:01}, also
known as normal factor graphs, are graphs that represent multivariate
functions. For example, let $\nset{T} = 4$ and $\nset{M} = 3$, let
$\set{T}_i$, $i = 1, \ldots, \nset{T}$ and $\set{M}_i$, $i = 1, \ldots,
\nset{M}$ be some arbitrary alphabets, and consider the function $f: \
\prod_{i=1}^{\nset{T}} \set{T}_i \times \prod_{i=1}^{\nset{M}} \set{M}_i \ \to
\ \mathbb{R}$ that represents the mapping
\begin{align*}
    &(\vt, \vm)
       \mapsto  
       f_1(t_1, m_1, m_3) \,
       f_2(t_2, m_1, m_2) \,
       f_3(t_3, t_4, m_2, m_3).
\end{align*}
Here, $f$ is called the global function and is the product of $\nset{F}$
functions $f_i$, $i = 1, \ldots, \nset{F}$ (here $\nset{F} = 3$), which are
called local functions. Whereas the argument set of the function $f$
encompasses $t_1$, $t_2$, $t_3$, $t_4$, $m_1$, $m_2$, and $m_3$, the function
$f_1$ has only $t_1$, $m_1$, and $m_3$ as arguments, $f_2$ has only $t_2$,
$m_1$, and $m_2$ as arguments, and $f_3$ has only $t_3$, $t_4$, $m_2$, and
$m_3$ as arguments.  Graphically, we represent this function decomposition as
follows, cf.~Fig.~\ref{fig:simple:ffg:1}~(left):
\begin{itemize}

\item For each local function we draw a function node (vertex).

\item For each variable we draw an half-edge or an edge.

\item If a variable appears as an argument in only one local function then we
  draw an half-edge that is connected to that local function. If a variable
  appears as an argument in two local functions then we draw an edge that
  connects these two local functions.\footnote{Global functions that contain
    variables that appear in more than two local functions can always replaced
    by essentially equivalent global functions where all variables appear as
    an argument in at most two local functions. E.g., $f(x_1, x_2, x_3, x_4) =
    f_1(x_1, x_2) \cdot f_2(x_1, x_3) \cdot f_3(x_1, x_4)$ can be replaced by
    the essentially equivalent $f'(x'_1, x''_1, x'''_1, x_2, x_3, x_4) =
    f_1(x'_1, x_2) \cdot f_2(x''_1, x_3) \cdot f_3(x'''_1, x_4) \cdot [x'_1 =
    x'_2 = x'_3]$. With this, $f$ being ``essentially equivalent'' to $f'$
    means that whenever $f(x'_1, x''_1, x'''_1, x_2, x_3, x_4)$ is nonzero
    then $f(x_1, x_2, x_3, x_4) = f(x'_1, x''_1, x'''_1, x_2, x_3, x_4)$ with
    $x_1 = x'_1 = x''_1 = x'''_1$.}

\end{itemize}

In our example the variables that are associated with half-edges are labeled
$T_i$, $i = 1, \ldots, \nset{T}$, whereas the variables that are associated
with edges are labeled $M_i$, $i = 1, \ldots, \nset{M}$. This distinction of
variable labels will be very helpful later on when we will dualize the global
function.

Interesting are global functions where the local function argument sets are
strict subsets of the global function argument set: the fewer arguments appear
in the local functions, the sparser the corresponding FFG will be.

Consider again the FFG in Figure~\ref{fig:simple:ffg:1}~(left) and let the
local functions represent indicator functions, i.e.,
\begin{align*}
  f_1(t_1, m_1, m_3)
    &\defeq
       \big[
         (t_1, m_1, m_3) \in \code{C}_1
       \big], \\
  f_2(t_2, m_1, m_2)
    &\defeq
       \big[
         (t_2, m_1, m_2) \in \code{C}_2
       \big], \\
  f_3(t_3, t_4, m_2, m_3)
    &\defeq
       \big[
         (t_3, t_4, m_2, m_3) \in \code{C}_3
       \big],
\end{align*}
for some codes $\code{C}_1 \subseteq \set{T}_1 \times \set{M}_1 \times
\set{M}_3$, $\code{C}_2 \subseteq \set{T}_2 \times \set{M}_1 \times
\set{M}_2$, and $\code{C}_3 \subseteq \set{T}_3 \times \set{T}_4 \times
\set{M}_2 \times \set{M}_3$. The restriction of the resulting global function
to the variables $\vt$, i.e., to the variables that are associated with the
half-edges, yields the function $[ \vt \in \code{C} ]$ with
\begin{align*}
  \code{C}
    = \left\{
          \vt \in \prod_{i=1}^{\nset{T}} \set{T}_i
        \ \left| \
          \begin{array}{c}
            \text{there exists an $\vm$} \\
            \text{such that $f(\vt, \vm) = 1$}
          \end{array}
        \right.
      \right\} \; .
\end{align*}
Clearly, if the sets $\set{T}_i$, $i = 1, \ldots, \nset{T}$, and $\set{M}_i$,
$i = 1, \ldots, \nset{M}$, are groups and the codes $\code{C}_i$, $i = 1,
\ldots, \nset{F}$ are group codes then $\code{C}$ is a group code.\footnote{A
  code is a group code if it is a subgroup of the direct product of the symbol
  alphabet groups. Note that a group code can be defined as the span of a list
  of suitably chosen vectors. Considering the group operation as ``addition'',
  group codes can also be seen to be additive codes, i.e., codes that are
  closed under addition.}

\section{Dualizing FFGs}
\label{sec:dualizing:ffgs:1}

In the context of stabilizer QECCs, dual codes (under the symplectic inner
product) play a very important role. The aim of this section is to start with
an FFG that represents the indicator function of some code and to construct an
FFG that represents the indicator function of the dual of that code. To that
end we will heavily use insights from~\cite[Section~VII]{Forney:01:1} on
Pontryagin duality theory in the context of FFGs, notably one of relatively
few results that hold for graphical models \emph{without} and \emph{with}
cycles.

\newpage

For the rest of the paper, we make the following definitions.

\begin{Definition}
  \label{def:qsc:ring:def:1} 

  Let $p$ be some prime. 

  For $i = 1, \ldots, \nset{T}$:

  \begin{itemize}

  \item We define $\set{T}_i$ to be the group $\Zp^2$ with vector addition
    modulo $p$ and we denote elements of $\set{T}_i$ by $t_i = (t_{X,i},
    t_{Z,i})$.

  \item We let $\set{T}'_i$ be the character group of $\set{T}_i$. Because
    $\set{T}'_i$ turns out to be isomorphic to $\set{T}_i$, we identify
    $\set{T}'_i$ with $\set{T}_i$. Elements of $\set{T}'_i$ will be denoted by
    $t'_i = (t'_{X,i}, t'_{Z,i})$. 

  \item We define the inner product $\innerprod{t_i}{t'_i}: \ \set{T}_i \times
    \set{T}'_i \to \Zp$ to be the symplectic inner product, i.e.,
    \begin{align*}
      \innerprod{t_i}{t'_i}
        &\defeq
           \begin{bmatrix}
             t_{X,i} \\
             t_{Z,i}
           \end{bmatrix}^\tr \!\!\!
           \cdot
           \begin{bmatrix}
             0 & 1 \\
             1 & 0
           \end{bmatrix}
          \cdot
          \begin{bmatrix}
             t'_{X,i} \\
             t'_{Z,i}
           \end{bmatrix}
         = t_{X,i} t'_{Z,i}
           +
           t_{Z,i} t'_{X,i} \, ,
    \end{align*}
    where ${}^\tr$ represents vector transposition and where addition and
    multiplication are modulo $p$. (Note that we are using angular brackets to
    denote inner products. This is in contrast to~\cite{Forney:01:1} that used
    angular brackets to denote pairings, which are exponential functions of
    inner products.)

  \end{itemize}

  For $i = 1, \ldots, \nset{M}$:

  \begin{itemize}

  \item We let $\mu_i$ be some positive integer, we define $\set{M}_i$ to be
    the group $(\Zp^2)^{\mu_i}$ with vector addition modulo $p$, and we denote
    elements of $\set{M}_i$ by $m_i = \bigl( (m_{X,i,1}, m_{Z,i,1}), \ldots,
    (m_{X,i,\mu_i}, m_{Z,i,\mu_i}) \bigr)$.

  \item We let $\set{M}'_i$ denote the character group of $\set{M}_i$. Again,
    because $\set{M}'_i$ turns out to be isomorphic to $\set{M}_i$, we
    identify $\set{M}'_i$ with $\set{M}_i$. Elements of $\set{M}'_i$ will be
    denoted by $m'_i = \bigl( (m'_{X,i,1}, m'_{Z,i,1}), \ldots,
    (m'_{X,i,\mu_i}, m'_{Z,i,\mu_i}) \bigr)$.

  \item We define the inner product $\innerprod{m_i}{m'_i}$ to be the
    symplectic inner product, i.e.,
    \begin{align*}
      &\innerprod{m_i}{m'_i} \\
        &\defeq
           \begin{bmatrix}
             m_{X,i,1} \\ m_{Z,i,1} \\
             \vdots \\
             m_{X,i,\mu_i} \\ m_{Z,i,\mu_i}
           \end{bmatrix}^\tr \!\!\!
           \cdot
           \begin{bmatrix}
             0 & 1 & & & \\
             1 & 0 & & & \\
               &   & \ddots & & & \\ 
               &   & & 0 & 1 \\
               &   & & 1 & 0
           \end{bmatrix}
           \cdot
           \begin{bmatrix}
             m'_{X,i,1} \\ m'_{Z,i,1} \\
             \vdots \\
             m'_{X,i,\mu_i} \\ m'_{Z,i,\mu_i}
           \end{bmatrix} \\
        &= \sum_{h=1'}^{\mu_i}
             \big(
               m_{X,i,h} m'_{Z,i,h}
               +
               m_{Z,i,h} m'_{X,i,h}
             \big) \; .
    \end{align*}

  \end{itemize}

  The above inner products induce inner products on vectors, e.g.,
  $\innerprod{\vt}{\vt'}:\ \prod_{i=1}^{\nset{T}} \set{T}_i \times
  \prod_{i=1}^{\nset{T}} \set{T}'_i \to \Zp$ is the inner product defined by
  $\innerprod{\vt}{\vt'} \defeq \sum_{i=1}^{\nset{T}} \innerprod{t_i}{t'_i}$,
  etc..

  Moreover, all codes are assumed to be group codes.
\end{Definition}

In the following, because of the natural isomorphism of the groups
$(\Zp^2)^n$ and $(\Zp^n)^2$, a vector $\vt \in (\Zp^2)^n$ will not
only be written as
\begin{align*}
  \vt
    &= \big( (t_{X,1},t_{Z,1}), \ldots, (t_{X,n},t_{Z,n}) \big)
\end{align*}
but also as
\begin{align*}
  \vt
     = (\vt_X, \vt_Z) \ \text{ with } \ 
 \vt_X
    &= (t_{X,1}, \ldots, t_{X,n}) \in \Zp^n \; , \\
 \vt_Z
    &= (t_{Z,1}, \ldots, t_{Z,n}) \in \Zp^n \; .
\end{align*}
With this convention, the symplectic inner product of $\vt$ and $\vt'$ can be
written as
\begin{align*}
  \innerprod{\vt}{\vt'}
    &\defeq
       \begin{bmatrix}
         t_{X,1} \\ t_{Z,1} \\
         \vdots \\
         t_{X,n} \\ t_{Z,n}
       \end{bmatrix}^\tr \!\!\!
       \cdot
       \begin{bmatrix}
         0 & 1 & & & \\
         1 & 0 & & & \\
           &   & \ddots & & & \\ 
           &   & & 0 & 1 \\
           &   & & 1 & 0
       \end{bmatrix}
       \cdot
       \begin{bmatrix}
         t'_{X,1} \\ t'_{Z,1} \\
         \vdots \\
         t'_{X,n} \\ t'_{Z,n}
       \end{bmatrix}
\end{align*}
or as
\begin{align*}
  \innerprod{\vt}{\vt'}
    &= \begin{bmatrix}
         \vt_X^\tr \\ \vt_Z^\tr
       \end{bmatrix}^\tr \!\!\!
       \cdot
       \begin{bmatrix}
         \matr{0}     & \mathbb{1}_n \\
         \mathbb{1}_n & \matr{0}
       \end{bmatrix}
       \cdot
       \begin{bmatrix}
         (\vt'_X)^\tr \\ (\vt'_Z)^\tr
       \end{bmatrix} \\
    &= \vt_X \cdot (\vt'_Z)^\tr
       +
       \vt_Z \cdot (\vt'_X)^\tr \; ,
\end{align*}
where $\mathbb{1}_n$ is the $n \times n$ identity matrix. Similar expressions
will also be used for the vector $\vm$ and combinations of $\vt$ and $\vm$.

\begin{Definition}
  \label{def:dual:code:1}
  
  The dual code $\code{C}^\perps$ (under the symplectic inner product) of some
  group code $\code{C} \subseteq \prod_{i=1}^{\nset{T}} \set{T}_i$ is defined
  to be
  \begin{align*}
    \code{C}^\perps
      &\defeq
         \left\{
           \left.
           \vt' \in \prod_{i=1}^{\nset{T}} \set{T}'_i
           \ \right| \
           \innerprod{\vt}{\vt'} = 0
         \right\} \; .
  \end{align*}
\end{Definition}

Note that $\code{C}^\perps$ is also a group code and that
$(\code{C}^\perp)^\perp = \code{C}$. Similarly, for any $i = 1, \ldots,
\nset{T}$, because $\code{C}_i$ was assumed to be a group code, we can define
its dual $\code{C}_i^\perps$.

In the following, we want to show that there is an FFG representing
$\code{C}^\perps$ that is tightly related to the FFG that represents
$\code{C}$. Continuing our example from Section~\ref{sec:ffgs:1}, let
$\set{M}''_i \defeq \set{M}'_i$ for $i = 1, \ldots, \nset{M}$, and let
$f^\perp: \ \prod_{i=1}^{\nset{T}} \set{T}'_i \times \prod_{i=1}^{\nset{M}}
\set{M}'_i \times \prod_{i=1}^{\nset{M}} \set{M}''_i \ \to \ \mathbb{R}$ be
the function that represents the mapping
\begin{align*}
    (\vt', \vm', \vm'')
     \mapsto  
    &f^\perp_1(t'_1, m'_1, m''_3) \ 
       f^\perp_2(t'_2, m''_1, m'_2) \\
    &  \cdot \,
       f^\perp_3(t'_3, t'_4, m''_2, m'_3) \\
    &  \cdot \,
       \big[ m'_1 {=} - m''_1 \big] \,
       \big[ m'_2 {=} - m''_2 \big] \,
       \big[ m'_3 {=} - m''_3 \big]
\end{align*}
with
\begin{align*}
  f^\perp_1(t'_1, m'_1, m''_3)
    &\defeq
       \big[
         (t'_1, m'_1, m''_3) \in \code{C}^\perps_1
       \big] \; , \\
  f^\perp_2(t'_2, m''_1, m'_2)
    &\defeq
       \big[
         (t'_2, m''_1, m'_2) \in \code{C}^\perps_2
       \big] \; , \\
  f^\perp_3(t'_3, t'_4, m''_2, m'_3)
    &\defeq
       \big[
         (t'_3, t'_4, m''_2, m'_3) \in \code{C}^\perps_3
       \big] \; .
\end{align*}
The function $f^\perp$ is depicted by the FFG in
Figure~\ref{fig:simple:ffg:1}~(middle) where the function nodes with a tilde
in them represent the indicator functions $[m'_1 {=} - m''_1]$, $[m'_2 {=} -
m''_2]$, and $[m'_3 {=} - m''_3]$, respectively. With this, we can
follow~\cite{Forney:01:1} and establish the next theorem.

\begin{Theorem}
  \label{th:code:duality:1}

  With the above definitions,
  \begin{align*}
    \code{C}^\perps 
      &= \left\{
           \left.
             \vt' \in \prod_{i=1}^{\nset{T}} \set{T}'_i
           \ \right| \
           \begin{array}{c}
             \text{there exists $\vm'$ and $\vm''$} \\
             \text{such that $f^\perp(\vt', \vm', \vm'') = 1$}
           \end{array}
         \right\} \; .
  \end{align*}
\end{Theorem}

\begin{proof}
  (The proof is for the example code in Figure~\ref{fig:simple:ffg:1}, but the
  proof can easily be generalized.) Let $\vt' \in \prod_{i=1}^{\nset{T}}
  \set{T}'_i$ be such that there exist $\vm'$ and $\vm''$ such that
  $f^\perp(\vt, \vm', \vm'') = 1$. Moreover, let $\vt \in \code{C}$ and let
  $\vm$ be such that $f(\vt, \vm) = 1$. Then,
  \begin{align*}
    \innerprod{\vt}{\vt'}
      &= \sum_{i=1}^{\nset{T}}
           \innerprod{t_i}{t'_i} \\
      &= \sum_{i=1}^{\nset{T}}
           \innerprod{t_i}{t'_i}
         +
         \sum_{i=1}^{\nset{M}}
           \innerprod{m_i}{m'_i}
         -
         \sum_{i=1}^{\nset{M}}
           \innerprod{m_i}{m'_i} \\
      &= \sum_{i=1}^{\nset{T}}
           \innerprod{t_i}{t'_i}
         +
         \sum_{i=1}^{\nset{M}}
           \innerprod{m_i}{m'_i}
         +
         \sum_{i=1}^{\nset{M}}
           \innerprod{m_i}{m''_i} \\
      &= \big(
           \innerprod{t_1}{t'_1}
           +
           \innerprod{m_1}{m'_1}
           + 
           \innerprod{m_3}{m''_3}
         \big) + \\
      &\quad\ 
         \big(
           \innerprod{t_2}{t'_2}
           +
           \innerprod{m_1}{m''_1}
           + 
           \innerprod{m_2}{m'_2}
         \big) + \\
      &\quad\ 
         \big(
           \innerprod{t_3}{t'_3}
           +
           \innerprod{t_4}{t'_4}
           +
           \innerprod{m_2}{m''_2}
           + 
           \innerprod{m_3}{m'_3}
         \big) \\
       &\overset{(*)}{=}
          0 + 0 + 0
        = 0 \; .
  \end{align*}
  Here, step $(*)$ follows from the fact that $(t_1, m_1, m_3) \in \code{C}_1$
  and $(t'_1, m'_1, m''_3) \in \code{C}^\perps_1$ imply that
  $\innerprod{t_1}{t'_1} + \innerprod{m_1}{m'_1} + \innerprod{m_3}{m''_3} =
  0$, with similar expressions for the other subcodes.

  We see that $\vt'$ is orthogonal to $\vt$, and because $\vt \in \code{C}$
  was arbitrary, $\vt'$ must be in $\code{C}^\perps$.
\end{proof}

\begin{Assumption}
  \label{ass:characteristic:two:1}

  For the rest of the paper we will assume that $p = 2$, which implies that
  the groups $\set{T}_i$, $\set{T}'_i$, $i = 1, \ldots, \nset{T}$, and the
  groups $\set{M}_i$, $\set{M}'_i$, $\set{M}''_i$, $i = 1, \ldots, \nset{M}$,
  have characteristic $2$. Therefore, $m''_i = - m''_i$ for all $m''_i \in
  \set{M}''_i$, $i = 1, \ldots, \nset{M}$, etc..
\end{Assumption}

So, given that we are only interested in arguments of $f^\perp$ that lead to
non-zero function values, any valid configuration $(\vt', \vm', \vm'')$ of the
FFG in Figure~\ref{fig:simple:ffg:1}~(middle) fulfills $\vm' = - \vm'' =
\vm''$. This observation allows us to simplify the function $f^\perp$ to
$f^\perp: \ \prod_{i=1}^{\nset{T}} \set{T}'_i \times \prod_{i=1}^{\nset{M}}
\set{M}'_i \ \to \ \mathbb{R}$ that represents the mapping
\begin{align*}
    &(\vt', \vm')
      \mapsto   
       f^\perp_1(t'_1, m'_1, m'_3)
       f^\perp_2(t'_2, m'_1, m'_2)
       f^\perp_3(t'_3, t'_4, m'_2, m'_3)
\end{align*}
with
\begin{align*}
  f^\perp_1(t'_1, m'_1, m'_3)
    &\defeq
       \big[
         (t'_1, m'_1, m'_3) \in \code{C}^\perps_1
       \big] \; , \\
  f^\perp_2(t'_2, m'_1, m'_2)
    &\defeq
       \big[
         (t'_2, m'_1, m'_2) \in \code{C}^\perps_2
       \big] \; , \\
  f^\perp_3(t'_3, t'_4, m'_2, m'_3)
    &\defeq
       \big[
         (t'_3, t'_4, m'_2, m'_3) \in \code{C}^\perps_3
       \big] \; .
\end{align*}
The new function $f^\perp$ is depicted by the FFG in
Figure~\ref{fig:simple:ffg:1}~(right). It is clear that
Theorem~\ref{th:code:duality:1} simplifies to the following corollary.

\begin{Corollary}
  \label{cor:code:duality:2}

  With the above definitions and Assumption~\ref{ass:characteristic:two:1} we
  have
  \begin{align*}
    \code{C}^\perps 
      &= \left\{
           \left.
             \vt' \in \prod_{i=1}^{\nset{T}} \set{T}'_i
           \ \right| \
           \begin{array}{c}
             \text{there exists an $\vm'$} \\
             \text{such that $f^\perp(\vt', \vm') = 1$}
           \end{array}
         \right\} \; .
  \end{align*}
\end{Corollary}

\begin{proof}
  Follows easily from Theorem~\ref{th:code:duality:1}.
\end{proof}

We conclude this section with a definition that will be crucial for the
remainder of this paper, namely self-orthogonality and self-duality (under the
symplectic inner product) of a code.

\begin{Definition}
  \label{def:self:orthogonal:code:1}

  Let $\code{C}$ be a group code with dual code $\code{C}^\perp$. Then,
  \begin{itemize}

  \item $\code{C}$ is called self-orthogonal if $\code{C} \subseteq
    \code{C}^{\perps}$ and

  \item $\code{C}$ is called self-dual if $\code{C} = \code{C}^\perps$.

  \end{itemize}
  (Note that a code $\code{C}$ is self-orthogonal if $\innerprod{\vt}{\vt'} =
  0$ for all $\vt, \vt' \in \code{C}$.)
\end{Definition}

\section{Stabilizer Label Codes and \\
             Normalizer Label Codes}
\label{sec:stabilizer:and:normalizer:label:codes:1}

Let $\code{C}$ be a code over $\Inttwo^2$ that is self-orthogonal under the
symplectic inner product. Without going into the details of the stabilizer
QECC framework, such a code $\code{C}$ can be used to construct a stabilizer
QECC. In that context, the codes $\code{C}$ and $\code{C}^\perp$ are called,
respectively, the stabilizer label code and the normalizer label code
associated with that stabilizer QECC.

\begin{Proposition}
  \label{prop:ffg:self:orthogonal:code:1}

  Using the notation that has been introduced so far, in particular
  Definition~\ref{def:qsc:ring:def:1} and
  Assumption~\ref{ass:characteristic:two:1}, let $\code{C} \subseteq
  \prod_{i=1}^{\nset{T}} \set{T}_i$ be a group code whose indicator function
  is defined by an FFG containing half-edges $T_i$, $i = 1, \ldots, \nset{T}$,
  full edges $M_i$, $i = 1, \ldots, \nset{M}$, and function nodes $f_i$, $i =
  1, \ldots, \nset{F}$, where the latter are indicator functions of group
  codes $\code{C}_i$, $i = 1, \ldots, \nset{F}$. Then,
  \begin{itemize}
  
  \item $\code{C}$ is self-orthogonal if all $\code{C}_i$ are self-orthogonal,
    and

  \item $\code{C}$ is self-dual if all $\code{C}_i$ are self-dual.

  \end{itemize}
\end{Proposition}

\begin{proof}
  First we consider the case where all $\code{C}_i$ are self-orthogonal. The
  code $\code{C}$ can be represented by an FFG like the FFG in
  Figure~\ref{fig:simple:ffg:1}~(left). Let $\vt$ be a codeword in $\code{C}$
  and let $\vm$ be such that $f(\vt, \vm) = 1$. Because of
  Definition~\ref{def:qsc:ring:def:1},
  Assumption~\ref{ass:characteristic:two:1}, and
  Corollary~\ref{cor:code:duality:2}, its dual code $\code{C}^\perps$ can be
  represented by an FFG like the FFG in Figure~\ref{fig:simple:ffg:1}~(right).
  Then, because the FFG in Figure~\ref{fig:simple:ffg:1}~(left) is
  topologically equivalent to the FFG in Figure~\ref{fig:simple:ffg:1}~(right)
  and because all $\code{C}_i$ are self-orthogonal, it follows that
  $f^\perp(\vt, \vm) = 1$, which in turn yields $\vt \in \code{C}^\perps$.
  Finally, because $\vt \in \code{C}$ was arbitrary, we see that $\code{C}
  \subseteq \code{C}^\perps$, i.e., that $\code{C}$ is self-orthogonal.

  Secondly, we consider the case where all $\code{C}_i$ are self-dual.
  Similarly to the above argument, we can show that $\code{C} \subseteq
  \code{C}^\perps$. Reversing the roles of $\code{C}$ and $\code{C}^\perps$, we
  can also show that $\code{C}^\perps \subseteq \code{C}$. This proves that
  $\code{C} = \code{C}^\perps$, i.e., that $\code{C}$ is self-dual.
\end{proof}

Obviously, Proposition~\ref{prop:ffg:self:orthogonal:code:1} gives us a simple
tool to construct stabilizer label codes and normalizer label codes. It does
not seem that duality results for FFGs, which are at the heart of
Proposition~\ref{prop:ffg:self:orthogonal:code:1}, have been leveraged before
to construct stabilizer QECCs.\footnote{While preparing this paper, we became
  aware of the recent paper~\cite{Wang:Wang:Du:Zeng:08:1} which also uses
  factor graphs and some type of duality results in the context of stabilizer
  QECCs. However, that paper does not give enough details for one to be able
  to judge its merits towards constructing stabilizer QECCs.}

\subsection{CSS Codes}

CSS codes are a family of stabilizer QECCs named after Calderbank, Shor, and
Steane (see e.g.~\cite{Nielsen:Chuang:00:1}). For these codes we will not use
our formalism, however, later on CSS codes can be used as component codes for
longer codes.

Let $\code{B}_1 \subseteq \Inttwo^n$ and $\code{B}_2 \subseteq \Inttwo^n$ be
two binary codes of length $n$ such that $\vb \cdot (\vb')^\tr = 0$ for all
$\vb \in \code{B}_1$ and $\vb' \in \code{B}_2$. Based on these two binary
codes, we define the stabilizer label code
\begin{align*}
  \code{C}
    &\defeq 
       \big\{
           \vt = (\vt_X, \vt_Z)
         \ \big| \
           \vt_X \in \code{B}_1 \ \text{and} \ 
           \vt_Z \in \code{B}_2
         \big.
       \big\}. 
\end{align*}
It can easily be seen that the code $\code{C}$ is self-orthogonal. Namely, for
any $\vt = (\vt_X, \vt_Z), \ \vt' = (\vt'_X, \vt'_Z) \in \code{C}$ we have
$\innerprod{\vt}{\vt'} = \vt_X \cdot (\vt'_Z)^\tr + \vt_Z \cdot (\vt'_X)^\tr =
0 + 0 = 0$, where $\vt_X \cdot (\vt'_Z)^\tr = 0$ follows from $\vt_X \in
\code{B}_1$ and $\vt'_Z \in \code{B}_2$, and where $\vt_Z \cdot (\vt'_X)^\tr =
\vt'_X \cdot (\vt_Z)^\tr = 0$ follows from $\vt'_X \in \code{B}_1$ and $\vt_Z
\in \code{B}_2$.

\begin{Example}
  \label{ex:steane:code:1}

  The so-called seven qubit Steane stabilizer QECC (see
  e.g.~\cite{Nielsen:Chuang:00:1}) is a CSS code where both $\code{B}_1$ and
  $\code{B}_2$ equal the $[7,3,4]$ binary simplex code, i.e., the code given
  by the rowspan of the matrix
  \begin{align*}
         \begin{bmatrix}
           0 & 0 & 0 & 1 & 1 & 1 & 1 \\
           0 & 1 & 1 & 0 & 0 & 1 & 1 \\
           1 & 0 & 1 & 0 & 1 & 0 & 1
         \end{bmatrix} \; .
  \end{align*}
\end{Example}

\subsection{Codes over $\GFfour$}

Let us associate with any vector $\vt = (\vt_X, \vt_Z) \in (\Inttwo^n)^2$ the
vector $\vt_{\GFfour} = (t_{\GFfour, 1}, \ldots, t_{\GFfour, n}) \in
\GFfour^n$ through the mapping\footnote{The definition of $\GFfour$ was given
  at the end of Section~\ref{sec:introduction:1}.}
\begin{align*}
  \vt_{\GFfour}
    &\defeq
       \gamma_{\Inttwo^2\to\GFfour}(\vt)
     \defeq
       \omega \, \vt_X
       +
       \oomega \, \vt_Z \; .
\end{align*}
Clearly, the mapping $\gamma_{\Inttwo^2\to\GFfour}$ is injective and
surjective and therefore bijective, and so there is a bijective mapping
between $\code{C}$ and $\code{C}_{\GFfour} \defeq
\gamma_{\Inttwo^2\to\GFfour}(\code{C})$, the latter being the image of
$\code{C}$ under the mapping $\gamma_{\Inttwo^2\to\GFfour}$. Because
$\code{C}$ was assumed to be a group/additive code, $\code{C}_{\GFfour}$ is
also a group/additive code.  Moreover, if for any codeword $\vt_{\GFfour} \in
\code{C}_{\GFfour}$ it holds that $\omega \cdot \vt_{\GFfour} \in
\code{C}_{\GFfour}$, then the code $\code{C}_{\GFfour}$ is a linear code,
i.e., not only is the sum of two codewords again a codeword, but any
$\GFfour$-multiple of a codewords is also a codeword. If $\code{C}_{\GFfour}$
is a linear code then also $\code{C}^\perps_{\GFfour}$ is a linear code. With
this, Proposition~\ref{prop:ffg:self:orthogonal:code:1} can be suitably
reformulated for sub-codes $\code{C}_{\GFfour, i} \defeq
\gamma_{\Inttwo^2\to\GFfour}(\code{C}_i)$ that are linear codes over
$\GFfour$, whereby one can show that the symplectic inner product can be
replaced by the Hermitian inner
product~\cite{Calderbank:Rains:Shor:Sloane:98:1}; we leave the details to the
reader. Note that a necessary condition for the code $\code{C}_{\GFfour}$ to
be linear is that $|\code{C}_{\GFfour}|$ is a power of $4$, i.e., that also
$|\code{C}|$ is a power of $4$.

\begin{Example}
  \label{ex:steane:code:2}

  The so-called five qubit stabilizer QECC (see e.g.
  \cite{Nielsen:Chuang:00:1}) has a stabilizer label code $\code{C}_{\GFfour}$
  that is the $\Inttwo$-rowspan of
  \begin{align*}
    \begin{bmatrix}
      \omega  & \oomega & \oomega & \omega  & 0 \\
      0       & \omega  & \oomega & \oomega & \omega \\
      \omega  & 0       & \omega  & \oomega & \oomega \\
      \oomega & \omega & 0       & \omega  & \oomega
    \end{bmatrix},
  \end{align*}
  It can easily be checked that $\code{C}_{\GFfour}$ is a linear code, which
  allows one to represent it as the $\GFfour$-rowspan of
  \begin{align*}
    \begin{bmatrix}
      \omega & \oomega & \oomega & \omega  & 0 \\
      0      & \omega  & \oomega & \oomega & \omega
    \end{bmatrix}.
  \end{align*}
\end{Example}

\section{Examples}
\label{sec:examples:1}

In this section we show how Proposition~\ref{prop:ffg:self:orthogonal:code:1}
can be leveraged to construct stabilizer label codes, in particular how that
proposition unifies several earlier proposed stabilizer label code
constructions. (For more details about the discussed codes we refer to the
corresponding papers.)

\begin{figure}
  \begin{center}
    \epsfig{file=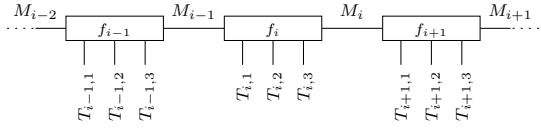, width=0.8\linewidth}
  \end{center}
  \caption{FFG for the convolutional stabilizer label codes in
    Exs.~\ref{ex:conv:stab:qecc:1} and~\ref{ex:conv:stab:qecc:2}.}
  \label{fig:ffg:example:forney:grassl:guha:1}
\end{figure}

\begin{Example}
  \label{ex:conv:stab:qecc:1}

  \textbf{(Convolutional
    Stab.~QECC~\cite[Example~1]{Forney:Grassl:Guha:07:1})} With the help of
  our FFG framework, the stabilizer label code
  of~\cite[Example~1]{Forney:Grassl:Guha:07:1} can be seen to be given by the
  FFG in Figure~\ref{fig:ffg:example:forney:grassl:guha:1}, where, using the
  notation from Definition~\ref{def:qsc:ring:def:1}, $\mu_i = 1$ for all $i$.
  For all $i$, the local function $f_i$ is given by
  \begin{align*}
    f_i(m_{i-1}, t_{i,1}, t_{i,2}, t_{i,3}, m_i)
      &\defeq
         \big[
           (m_{i-1}, t_{i,1}, t_{i,2}, t_{i,3}, m_i) \in \code{C}_i
         \big]
  \end{align*}
  with $\code{C}_i$ such that $\code{C}_{\GFfour, i} \defeq
  \gamma_{\Inttwo^2\to\GFfour}(\code{C}_i)$ is a linear code that is the
  $\GFfour$-rowspan of the matrix
  \begin{align*}
    \left[
      \begin{array}{c|ccc|c}
        0 & 1 & 1      & 1       & 1 \\
        1 & 1 & \omega & \oomega & 0
      \end{array}
    \right] \; .
  \end{align*}
  (In order to obtain a block code one needs to terminate the FFG in
  Figure~\ref{fig:ffg:example:forney:grassl:guha:1} on both sides; we omit the
  discussion of this issue. Alternatively, tail-biting can be used.)
\end{Example}

\begin{Example}
  \label{ex:conv:stab:qecc:2}

  \textbf{(Convolutional
    Stab.~QECC~\cite[Example~3]{Forney:Grassl:Guha:07:1})} Similarly, we can
  represent the stabilizer label code
  of~\cite[Example~3]{Forney:Grassl:Guha:07:1} by the FFG in
  Figure~\ref{fig:ffg:example:forney:grassl:guha:1}. Here, however, we have
  $\mu_i = 2$ for all $i$, and $\code{C}_i$ is such that $\code{C}_{\GFfour,
    i} \defeq \gamma_{\Inttwo^2\to\GFfour}(\code{C}_i)$ is a linear code given
  by the $\GFfour$-rowspan of
  \begin{align*}
    \left[
      \begin{array}{cc|ccc|cc}
        0 & 0 & 1 & 1 & 1 & 0 & 1 \\
        0 & 1 & 1 & 0 & 0 & 1 & 1 \\
        1 & 1 & 1 & 1 & 0 & 0 & 0
      \end{array}
    \right] \; .
  \end{align*}
  Note that a ``more common'' choice for a matrix whose $\GFfour$-rowspan is
  the trellis section code of a non-recursive convolutional code would have
  been a matrix like
  \begin{align*}
    \left[
      \begin{array}{cc|ccc|cc}
        0 & 0 & 1 & 1 & 1 & 0 & 1 \\
        0 & 1 & 1 & 0 & 0 & 1 & 0 \\
        1 & 0 & 1 & 1 & 0 & 0 & 0
      \end{array}
    \right] \; .
  \end{align*}
  Here the rows are such that the last $\mu_i{-}1$ components of $m_{i-1}$
  equal the first $\mu_i{-}1$ components of $m_i$. However, the $\GFfour$-span
  of such a matrix does not result in a self-orthogonal $\code{C}_i$.
\end{Example}

\begin{Example}
  \label{ex:serial:turbo:stab:qecc:1}

  \textbf{(Serial Turbo
    Stab.~QECC~\cite[Figure~10]{Poulin:Tillich:Ollivier:07:1:subm})} The
  paper~\cite{Poulin:Tillich:Ollivier:07:1:subm} discusses constructions of
  serial turbo stabilizer label codes. In particular, Figure~10
  in~\cite{Poulin:Tillich:Ollivier:07:1:subm} presents a code that corresponds
  to the FFG shown in
  Figure~\ref{fig:ffg:example:poulin:tillich:ollivier:1}~(left). Here,
  $f_{\mathrm{convcode}1}$, $f_{\mathrm{quantum-interleaver}}$, and
  $f_{\mathrm{convcode}2}$ represent, respectively, the indicator functions of
  the first convolutional code, of the quantum interleaver, and of the second
  convolutional code. If the indicator functions correspond to self-orthogonal
  codes then we can apply our FFG framework and guarantee that the overall
  code is self-orthogonal. This is indeed the case for the codes presented
  in~\cite{Poulin:Tillich:Ollivier:07:1:subm}.

  \begin{figure}
    \begin{center}
      \epsfig{file=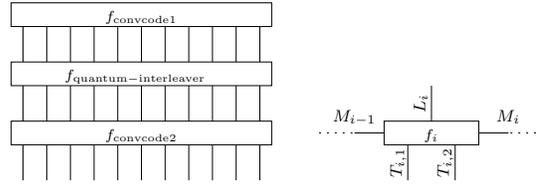, width=0.79\linewidth}
    \end{center}
    \caption{Left: FFG for the serial turbo stabilizer label code in
      Example~\ref{ex:serial:turbo:stab:qecc:1}. Right: FFG for the second
      convolutional code on the left-hand side.}
    \label{fig:ffg:example:poulin:tillich:ollivier:1}
  \end{figure}

  A particular example of a stabilizer label code that can be used for the
  function node $f_{\mathrm{convcode}2}$ is given in~\cite[Figures~8
  and~9]{Poulin:Tillich:Ollivier:07:1:subm} and shown as an FFG in
  Figure~\ref{fig:ffg:example:poulin:tillich:ollivier:1}~(right). (In contrast
  to~\cite[Figures~8 and~9]{Poulin:Tillich:Ollivier:07:1:subm}, that uses the
  variable names $P_{i,1}$ and $P_{i,2}$, we are using the variable names
  $T_{i,1}$ and $T_{i,2}$, respectively.) Here, for all $i$ the indicator
  function $f_i$ corresponds to a code $\code{C}_i$ that is the rowspan of
  \begin{align*}
    \left[
      \begin{array}{ccccc|ccccc}
        1 & 0 & 1 & 1 & 1 &  0 & 0 & 0 & 0 & 0 \\     
        0 & 1 & 0 & 1 & 1 &  0 & 0 & 0 & 0 & 0 \\     
        0 & 0 & 0 & 0 & 0 &  1 & 0 & 1 & 0 & 0 \\     
        0 & 0 & 0 & 0 & 0 &  0 & 1 & 1 & 1 & 0 \\     
        0 & 0 & 0 & 0 & 0 &  0 & 0 & 0 & 1 & 1     
      \end{array}
    \right] \; ,
  \end{align*}
  where the columns correspond to $m_{X, i-1}$, $l_{X,i}$, $t_{X,i,1}$,
  $t_{X,i,2}$, $m_{X,i}$, $m_{Z,i-1}$, $l_{Z,i}$, $t_{Z,i,1}$, $t_{Z,i,2}$,
  $m_{Z,i}$, respectively. Note that $\code{C}_i$ is self-dual under the
  symplectic inner product and that the corresponding $\GFfour$-code
  $\code{C}_{\GFfour, i} \defeq \gamma_{\Inttwo^2\to\GFfour}(\code{C}_i)$ is
  additive but not linear. (It cannot be linear since the number of codewords
  is $32$, which is not a power of $4$.)
\end{Example}

\begin{Example}
  \label{ex:stab:state:1}

  \textbf{(Stabilizer State~\cite[Example~1]{VandenNest:Dehaene:DeMoor:04:1})}
  Roughly speaking, a stabilizer state corresponds to a stabilizer QECC whose
  stabilizer label code is self-dual, see \cite{Schlingemann:Werner:01:1,
    Schlingemann:01:1, VandenNest:Dehaene:DeMoor:04:1,
    Cross:Smith:Smolin:Zeng;07:1:subm}. Let $\matr{A}$ be the $n \times n$
  adjacency matrix of any graph with $n$ vertices and let $\code{C}$ be the
  rowspan of $\bigl[ \, \mathbb{1}_n \, \bigl| \, \matr{A} \, \bigr.  \bigr]$,
  where the columns correspond to $t_{X,1}, \ldots, t_{X,n}$, $t_{Z,1},
  \ldots, t_{Z,n}$, and where $\mathbb{1}_n$ is the $n \times n$ identity
  matrix. It can easily be checked that $\code{C}$ is self-dual under the
  symplectic inner product. (Note that $\matr{A}^\tr = \matr{A}$ because
  $\matr{A}$ is the adjacency matrix of a graph.)

  \begin{figure}
    \begin{center}
      \epsfig{file=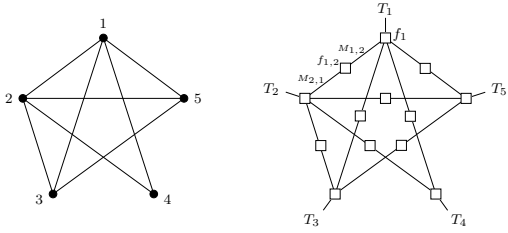, width=0.75\linewidth}
    \end{center}
    \caption{Left: graph with five vertices that defines a stabilizer state.
      Right: FFG of the corresponding stabilizer label code.}
    \label{fig:ffg:example:vandennest:dehaene:demoor:1}
  \end{figure}

  For example, the graph in
  Figure~\ref{fig:ffg:example:vandennest:dehaene:demoor:1}~(left) results in a
  stabilizer label code $\code{C}$ which is the rowspan of
  \begin{align*}
    \left[
      \begin{array}{ccccc|ccccc}
        1 & 0 & 0 & 0 & 0 &  0 & 1 & 1 & 1 & 1 \\     
        0 & 1 & 0 & 0 & 0 &  1 & 0 & 1 & 1 & 1 \\     
        0 & 0 & 1 & 0 & 0 &  1 & 1 & 0 & 0 & 1 \\     
        0 & 0 & 0 & 1 & 0 &  1 & 1 & 0 & 0 & 0 \\     
        0 & 0 & 0 & 0 & 1 &  1 & 1 & 1 & 0 & 0     
      \end{array}
    \right] \; .
  \end{align*}

  It turns out that any such stabilizer label code can also be represented by
  an FFG that is topographically closely related to the graph that defined the
  code. E.g., Figure~\ref{fig:ffg:example:vandennest:dehaene:demoor:1}~(right)
  shows the FFG that corresponds to the example graph in
  Figure~\ref{fig:ffg:example:vandennest:dehaene:demoor:1}~(left). Here, $f_1$
  is the indicator function of the self-dual code $\code{C}_1$ which is
  defined to be the rowspan of
  \begin{align*}
    \left[
      \begin{array}{ccccc|ccccc}
        1 & 0 & 0 & 0 & 0 &  0 & 1 & 1 & 1 & 1 \\     
        0 & 1 & 0 & 0 & 0 &  1 & 0 & 0 & 0 & 0 \\     
        0 & 0 & 1 & 0 & 0 &  1 & 0 & 0 & 0 & 0 \\     
        0 & 0 & 0 & 1 & 0 &  1 & 0 & 0 & 0 & 0 \\     
        0 & 0 & 0 & 0 & 1 &  1 & 0 & 0 & 0 & 0     
      \end{array}
    \right] \; ,
  \end{align*}
  where the columns correspond to $t_{X,1}$, $m_{X,1,2}$, $m_{X,1,3}$,
  $m_{X,1,4}$, $m_{X,1,5}$, $t_{Z,1}$, $m_{Z,1,2}$, $m_{Z,1,3}$, $m_{Z,1,4}$,
  $m_{Z,1,5}$, respectively. Moreover, $f_{1,2}$ is the indicator function of
  the self-dual code $\code{C}_{1,2}$ which is defined to be the rowspan of
  $\left[ \begin{array}{cc|cc} 1 & 0 & 0 & 1 \\ 0 & 1 & 1 & 0 \end{array}
  \right]$, where the columns correspond to $m_{X,1,2}$, $m_{X,2,1}$,
  $m_{Z,1,2}$, $m_{Z,2,1}$, respectively. The other indicator functions $f_i$
  and $f_{i,j}$ and self-dual codes $\code{C}_i$ and $\code{C}_{i,j}$ are
  defined analogously. Note that the code $\code{C}_i$ depends on the number
  of vertices that are adjacent to vertex $i$. However, the code
  $\code{C}_{i,j}$ is always the same for any pair $(i,j)$ of adjacent
  vertices.
\end{Example}

\begin{Example}
  \label{ex:pg:code:1}

  \textbf{(LDPC codes)} Any LDPC code whose parity-check matrix contains
  orthogonal rows can be used to construct a normalizer label code
  $\code{C}^\perp$, see e.g.~\cite{MacKay:Mitchison:McFadden:04:1,
    Aly:07:1:subm, Djordjevic:08:1} and references therein. In terms of FFGs,
  such LDPC codes are expressed with the help of equal and single-parity-check
  function nodes. The FFG of the corresponding stabilizer label code
  $\code{C}$ can also be expressed in terms of equal and single-parity-check
  function nodes. Because single-parity-checks of length not equal to $2$ do
  \emph{not} represent self-orthogonal codes, our FFG framework is not
  directly applicable to construct FFGs of such LDPC codes. However, with the
  help of some auxiliary code constructions, our framework can also be used to
  construct LDPC stabilizer/normalizer label codes; because of space
  constraints we do not give the details here.
\end{Example}

We leave it as an open problem to use our FFG framework to construct other
classes of stabilizer label codes that have interesting properties.

\section{Message-Passing Iterative and \\
             Linear Programming Decoding}
\label{sec:mpi:and:lp:decoding:1}

One of the main interests in studying FFGs for stabilizer label codes and
their duals is that one would like to have FFGs that are suitable for MPI
decoding. (Note that the code that is relevant for decoding in the stabilizer
QECC framework is the code $\code{C}^\perp$, or the coset code $\code{C}^\perp
/ \code{C}$, see e.g.~the comments
in~\cite{Poulin:Tillich:Ollivier:07:1:subm}.) The well-known trade-offs from
classical LDPC and turbo codes apply also here: good codes with low FFG
variable and function node complexity must have cycles, yet cycles lead to
sub-optimal performance of message-passing iterative decoders. We leave it as
an open problem to study stopping sets, trapping sets, absorbing sets,
near-codewords, pseudo-codewords, the fundamental polytope, etc.\ (see e.g.\
the refs.\ at \cite{Pseudocodewords:Website:05:1}) for the codes that were
discussed in this paper. Moreover, one can formulate alternative decoders to
MPI decoders like linear programming (LP) decoding. It would be interesting to
see if the self-orthogonality property of stabilizer label codes leads to
further insights in the context of MPI and LP decoders, in particular by also
leveraging other duality results for FFGs like Fourier
duality~\cite{Mao:Kschischang:05:1} and Lagrange
duality~\cite{Vontobel:Loeliger:02:2}.

\end{document}